\documentclass[11pt]{article}
\usepackage[margin=1in]{geometry}
\usepackage{amsmath,amssymb,amsthm}
\usepackage{graphicx}
\usepackage{natbib}
\usepackage{hyperref}
\usepackage{setspace}

\newtheorem{proposition}{Proposition}
\newtheorem{corollary}{Corollary}

\newtheorem{observation}{Observation}
\newtheorem{remark}{Remark}
\theoremstyle{definition}
\newtheorem{exercise}{Exercise}

\title{On the Optimality of Mental Accounting}
\author{Jarrod Burgh\\Grinnell \and Emerson Melo\\Indiana \and Matthew Shum\\Caltech, Amazon}
\date{\today}

\begin{document}
	\maketitle
	
	\begin{abstract}
		\noindent Mental accounting is usually explained by self-control, and two-stage budgeting by separable preferences. In this paper, we propose a third explanation based on sellers' market power. A consumer who commits to spend a fixed number of dollars in a category has unit-elastic demand, so a seller in that category gains nothing by restricting quantity. A consumer who instead targets a fixed physical quantity has perfectly inelastic demand, a case that a monopolist can exploit. We characterize the commitments that remove a seller's incentive to mark up prices, and show that the fixed dollar budget is the least restrictive. We then show that what matters is budget rigidity rather than separability of preferences. A consumer who re-optimizes her category budget when prices change is insulated from market power only if categories are substitutes; a rigid budget insulates her in every case. %Under a CES aggregator, markup-proofness holds if and only if the elasticity of substitution is at least one, and when categories are complements, the seller's revenue rises with price toward the consumer's entire income. We then give the condition under which a consumer chooses to commit, and show that the incentive to budget is strongest where competition is weakest. 
        Mental accounting can therefore be read as a commitment device against markups, a rationale that is separate from the self-control motive in the existing literature.
	\end{abstract}
	
	\section{Introduction}
	
	Consumers often divide their spending into category budgets and treat those budgets as binding. They decline to move money from the ``dining out'' envelope to the ``electronics'' envelope even when doing so would make them better off.
	
	Two strands of literature explain this behavior. The first, due to \citet{thaler1985mental} and developed in \citet{shefrin1988behavioral} and \citet{thaler1999mental}, treats mental accounts as a bookkeeping device that helps with self-control. The second, following \citet{gorman1959separable} and \citet{deaton1980economics}, shows that weakly separable preferences let the consumer's problem be solved in two stages: allocate income across groups, then optimize within each group given that group's budget.
	
Both literatures focus on what mental accounts do to the consumer's cognition or preferences. We ask a different question: \emph{what does a dollar budget do to the market the consumer faces?}
	
	The answer turns on the shape of demand that a commitment produces. A seller facing a buyer with a fixed dollar budget cannot profit from restricting supply. In this scenario, the buyer's demand is unit-elastic, so her total spending does not vary with price. On the other hand, a seller facing a buyer with a fixed physical-quantity target faces perfectly inelastic demand and profits by raising the price to the buyer's reservation value. The form of the commitment, dollars or units, determines whether the seller has any power over the buyer.

	We develop this observation in six results. We first characterize the full class of spending commitments that leave a seller with no profitable markup: a commitment has this property---we call it \emph{markup-proof}---if and only if the buyer's expenditure is nonincreasing in price, equivalently if and only if her demand is everywhere weakly elastic (Proposition~\ref{prop:characterization}). Furthermore, among all markup-proof commitments, the fixed dollar budget is the least restrictive, delivering weakly more consumption at every price than any alternative with the same property (Proposition~\ref{prop:minimal}). This characterization is the budget-commitment counterpart of the buyer-optimal demand of \citet{condorelli2020information}, in which a buyer who chooses her value distribution also ends up with unit-elastic demand; what is new here is that the commitment takes the form of a dollar budget inside a multi-category consumer problem. What insulates the consumer from market power is the rigidity of the budget rather than separability of preferences: a consumer who re-optimizes her group expenditure as prices change is markup-proof only when groups are substitutes, which under a CES upper stage means an elasticity of substitution of at least one; when groups are complements, her expenditure on a group \emph{rises} with its price and the seller's revenue rises toward her entire income (Proposition~\ref{prop:rigidity}). When the commitment is endogenous, the consumer prefers to budget when the markup she avoids exceeds the value of the flexibility she gives up (Proposition~\ref{prop:optimality}), and this incentive is strongest where competition is weakest, disappearing as the category becomes competitive (Proposition~\ref{prop:concentration}). The same incentive combines additively with a cost of tracking prices: a consumer who finds attention costly needs less market power to justify committing, and vice versa (Proposition~\ref{prop:attention}).
	
	Taken together, these results give a market-structure rationale for mental accounting and two-stage budgeting. Committing to a category-level dollar budget turns a consumer who would present quantity-target demand, which a seller can mark up, into one who presents unit-elastic demand, which a seller cannot. The argument requires no self-control problem and no present bias. It applies to a fully patient, dynamically consistent consumer.

	The paper proceeds as follows. Section~\ref{sec:background} presents the pricing mechanism, characterizes markup-proof commitments, and reinterprets mental accounting. Section~\ref{sec:gorman} builds on the Gorman--Deaton foundation and shows that rigidity, rather than separability, confers markup-proofness. Section~\ref{sec:optimality} makes the budget choice endogenous, derives a comparative statics result on market concentration, and shows how a cost of attention interacts with market power in the decision to commit. Section~\ref{sec:related} discusses the literature, and Section~\ref{sec:discussion} concludes. Proofs are in Appendix~\ref{app:proofs}; supporting derivations, including extensions to many buyers and qualifications on when budgeting fails, are in Appendix~\ref{app:derivations}.
	
	\section{Budget commitment versus quantity commitment}\label{sec:background}
	
	Consider a seller  of a good with zero marginal cost and fixed supply $s>0$, facing a single buyer with a constant per-unit value $v$.  
	
	\paragraph{A budget-capped buyer.} Suppose the buyer has a fixed dollar budget $b$ to spend in this market. She buys as many units as her budget allows whenever $p \le v$, so her demand is
	\[
	q(p) = \frac{b}{p},
	\]
	a rectangular hyperbola with elasticity $\varepsilon(p)=1$ at every $p \le v$. Revenue is $p\,q(p) = b$, which does not vary with price. The seller never prices below the market-clearing price $p^c = b/s$, and earns the same revenue at every price in $[p^c, v]$. Restricting quantity below supply yields him no gain.
	
	\paragraph{A quantity-capped buyer.} Suppose instead the buyer wants a fixed quantity $S$ and has reservation value $c$. Her demand is the step function
	\[
	q(p) = S\cdot\mathbf 1\{p<c\},
	\]
	so $\varepsilon(p)=0$ for $p<c$. Revenue $p\,S$ rises linearly with price up to $c$, so the seller raises the price to the buyer's reservation value.
	
	\begin{figure}[h]
		\centering
		\includegraphics[width=0.97\linewidth]{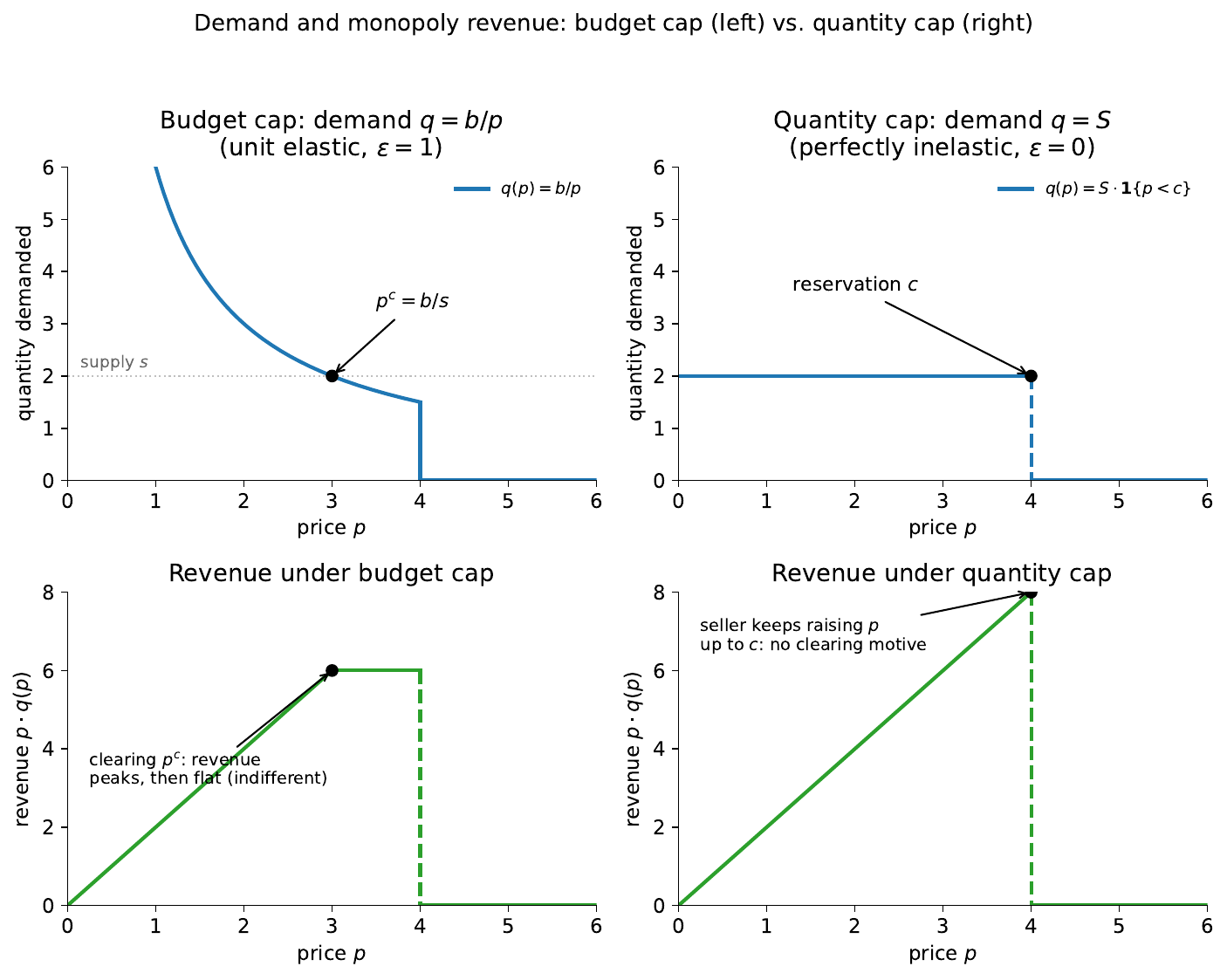}
		\caption{\footnotesize Demand (top) and monopoly revenue (bottom) for a single buyer under a budget cap (left) and a quantity cap (right), with the same reservation price and the same quantity and spend at the reference price. Budget-capped revenue peaks at the clearing price and stays flat up to the reservation. Quantity-capped revenue keeps rising to the reservation, so the seller has no reason to clear the market. See Appendix~\ref{app:single-buyer} for details.}
		\label{fig:caps-comparison}
	\end{figure}
	
	\subsection{Which commitments are markup-proof}\label{sec:charact}
	
	Both cases are instances of a more general principle. Let the buyer \emph{commit} to a \emph{spending rule}: a nonincreasing demand function $q:\mathbb R_+\to\mathbb R_+$ with $q(p)=0$ for $p>v$, describing what she buys at each posted price. Write $e(p) = p\,q(p)$ for her expenditure and $p^c(s) = \min\{p : q(p)\le s\}$ for the price that clears supply $s$. The seller posts a price $p$ to maximize $R(p) = p\cdot\min\{q(p),\,s\}$.
	
	Call the spending rule $q$ \emph{markup-proof} if for every supply level $s>0$ the clearing price $p^c(s)$ maximizes the seller's revenue, so that the seller never gains by holding quantity below capacity. The following result characterizes the class of 
	markup-proof rules.
	\begin{proposition}[Which commitments are markup-proof]\label{prop:characterization}
		Let $q$ be a spending rule that is continuous on $\{p : q(p)>0\}$. Then $q$ is markup-proof if and only if $e(p) = p\,q(p)$ is nonincreasing on $[0,v]$, equivalently if and only if $\varepsilon(p)\ge 1$ at every price at which $q$ is differentiable.
	\end{proposition}
	
	Continuity is a regularity condition. It holds for the budget rule and for the mixed rule of Appendix~\ref{app:mixed}, and it cannot be dispensed with.\footnote{A spending rule with a downward jump yields the seller strictly higher revenue just below the jump than at the clearing price. The quantity-target rule is the leading example: it is discontinuous at its reservation $c$, and at any supply $s<S$ its clearing price is $c$ itself, at which the buyer has already withdrawn and revenue is zero, whereas posting slightly less than $c$ yields approximately $cs$. Such rules can be accommodated by defining the seller's problem over left limits, at the cost of additional notation that the argument does not require.}
	
	The budget rule $q(p)=b/p$ satisfies the characterization with $e$ exactly constant. The quantity rule, being discontinuous, lies outside the scope of the proposition, but it fails to be markup-proof directly, as the footnote  1 indicates.
	
	The economic mechanism behind the characterization is the link between elasticity and marginal revenue. For a seller with zero marginal cost,
	\[
	\frac{d}{dp}\bigl[p\,q(p)\bigr] = q(p)\bigl(1-\varepsilon(p)\bigr),
	\]
	so a price increase raises revenue exactly when $\varepsilon<1$. Budget-capped and quantity-capped demand can agree at a single reference price: a budget $b$ and a quantity $S = b/p$ coincide there. They diverge once price moves away from that point, and it is this divergence, not the level of demand, that determines whether a markup pays.
	
	\begin{corollary}[No allocative distortion]\label{cor:nodwl}
		Under a markup-proof spending rule the seller's revenue-maximizing quantity equals the competitive quantity, so the deadweight loss usually associated with monopoly distortion is zero. Under the budget rule the seller's revenue is $b$ at every price in $[p^c,v]$.
	\end{corollary}
	
	A markup-proof rule therefore does more than transfer surplus from seller to buyer; it removes the quantity distortion entirely.
	
	Many rules are markup-proof; however, not all of them are equally advantageous for the buyer. Among rules that satisfy the same affordability constraint, the budget rule occupies a distinguished position.

	\begin{proposition}[The budget rule is the least restrictive markup-proof commitment]\label{prop:minimal}
		Fix $b>0$ and consider all markup-proof spending rules satisfying $e(p)\le b$ for all $p$. The budget rule $q(p)=b/p$ is pointwise maximal among them. It delivers weakly more consumption at every price $p\le v$ than any other markup-proof, affordable rule, and it is the unique rule attaining this bound.
	\end{proposition}
	
	The buyer therefore faces no trade-off between markup-proofing and consumption at the point of purchase. She could commit in other ways, for instance by undertaking to withdraw above some price, but each alternative secures the same immunity only at the cost of lower consumption. Committing to spend a fixed number of dollars is the weakest commitment that removes the seller's markup entirely, which is a reason to expect budgeting to take that form in practice.
	
	\subsection{Mental accounting as a commitment against markups}\label{sec:mental}
	
	Propositions~\ref{prop:characterization} and~\ref{prop:minimal} characterize a buyer solely in terms of her spending rule. We now specialize that rule to the two behavioral assumptions at issue, and ask what form a consumer's within-category demand takes under each.
	
	Following \cite{gorman1959separable}, we consider a  consumer with income $m>0$ to allocate across $K$ spending categories, each served by sellers with some degree of market power.
	
	\paragraph{Without mental accounts.} The consumer treats income as \emph{fully fungible} and forms a target quantity in each category: a number of restaurant meals per month, a certain wardrobe, etc. Facing a seller in category $k$, she is a \emph{quantity-capped buyer}, and her demand is close to perfectly inelastic up to her reservation value. Her spending rule is not markup-proof, so any seller with market power can raise price toward that reservation.
	
	\paragraph{With mental accounts.} The consumer partitions her income into category budgets $b_1,\ldots,b_K$, with $\sum_{k=1}^K b_k = m$, and treats each as a hard limit, independent of what she could reallocate from elsewhere. Facing a seller in category $k$, she is a \emph{budget-capped buyer}, and her demand is the unit-elastic hyperbola $q_k(p) = b_k/p$. Her spending rule is markup-proof, and by Proposition~\ref{prop:minimal} it is the least restrictive rule with that property.
	
	\begin{observation}\label{obs:mental}
		A consumer who commits to category-level dollar budgets converts her within-category demand from perfectly inelastic to unit-elastic. By Propositions~\ref{prop:characterization} and~\ref{prop:minimal}, this removes any seller's incentive to restrict supply and mark up price within that category, and does so at the least cost in foregone consumption among all commitments with that property.
	\end{observation}
	
	The familiar cost of mental accounting is misallocation. A rigid budget can leave money unspent in one category while a more valuable opportunity in another goes unfunded, which is precisely the fungibility violation that \citet{thaler1985mental} documents and that standard theory treats as an error. The observation above is that this cost purchases something in return: within each category, the consumer denies the seller the responsiveness he would require in order to raise price profitably. Section~\ref{sec:optimality} states the trade-off explicitly and gives the condition under which it favors budgeting.
	
	This rationale differs from the self-control rationale. A perfectly patient, dynamically consistent consumer facing no risk of temptation still has reason to adopt category budgets, provided she expects to face sellers with market power. The two rationales are complementary, and either one alone suffices to rationalize the behavior.
	
	\section{Rigidity versus separability}\label{sec:gorman}
	
	Two-stage budgeting appears to offer a fully rational route to the same object---a group-level dollar budget---without invoking any behavioral departure from fungibility. This section shows that separability generates budget but not necessarily markup-proofness, which depends  on whether the budget responds to prices.
	
	If a consumer's preferences over $K$ groups of goods are weakly separable,  \citet{gorman1959separable} and \citet{deaton1980economics} show that her problem decomposes. She first allocates her income $m$ across groups. She then chooses her within-group bundle to maximize a group-specific subutility, subject to spending exactly the amount allocated. Separability makes the group budget a sufficient statistic for within-group behavior: the composition of her purchases within group $k$ depends on prices elsewhere only through the amount she has allocated to group $k$.
	
	Separability alone does not make the budget markup-proof, however. The reason is that the first-stage allocation is itself a function of prices. Writing $P_k$ for the group-$k$ price index, her expenditure on group $k$ is
	\[
	e_k = E_k(P_1,\ldots,P_K,m).
	\]
	A rational consumer re-optimizes her budget when prices move, so her group expenditure need not be invariant to the price a within-group seller posts. Proposition~\ref{prop:characterization} requires expenditure to be nonincreasing in price, which is a restriction on how the first-stage allocation responds.
	
	\subsection{A CES upper-stage}\label{sec:ces}
	
	To determine what that restriction amounts to, we specialize the first stage. Let the upper-stage aggregator be CES,
	\begin{equation}\label{eq:ces}
		U(Q_1,\ldots,Q_K) \;=\; \Bigl(\textstyle\sum_{k=1}^{K} \alpha_k\,Q_k^{(\sigma-1)/\sigma}\Bigr)^{\sigma/(\sigma-1)},
		\quad \alpha_k>0\quad\text{for $k=1,\ldots,K,$}\ \sigma>0,
	\end{equation}
	where $Q_k$ is the group-$k$ composite and $\sigma$ is the elasticity of substitution across groups. Standard CES demand gives
	\begin{equation}\label{eq:cesdemand}
		Q_k \;=\; \frac{\alpha_k^{\sigma}P_k^{-\sigma}}{P^{1-\sigma}}\,m,
		\qquad
		P \;\equiv\; \Bigl(\textstyle\sum_{j}\alpha_j^{\sigma}P_j^{1-\sigma}\Bigr)^{1/(1-\sigma)} .
	\end{equation}
	Group expenditure is $e_k = P_kQ_k = w_k m$, where the budget share is
	\begin{equation}\label{eq:share}
		w_k \;=\; \frac{\alpha_k^{\sigma}P_k^{1-\sigma}}{\sum_{j}\alpha_j^{\sigma}P_j^{1-\sigma}} .
	\end{equation}
	
	Write $x_k \equiv \alpha_k^{\sigma}P_k^{1-\sigma}$ for group $k$'s term in the denominator of \eqref{eq:share} and $A_{-k}\equiv\sum_{j\ne k}\alpha_j^{\sigma}P_j^{1-\sigma}$ for the remaining terms. Holding the other groups' prices fixed, group-$k$ expenditure as a function of the seller's price is
	\begin{equation}\label{eq:bk}
		e_k(P_k) \;=\; m\,\frac{x_k}{x_k+A_{-k}} \;=\; m\,\frac{\alpha_k^{\sigma}P_k^{1-\sigma}}{\alpha_k^{\sigma}P_k^{1-\sigma}+A_{-k}} .
	\end{equation}
	Since $\partial \ln x_k/\partial \ln P_k = 1-\sigma$, differentiating \eqref{eq:share} gives the share elasticity
	\begin{equation}\label{eq:shareelast}
		\frac{\partial \ln w_k}{\partial \ln P_k} \;=\; (1-\sigma)(1-w_k),
	\end{equation}
	and because $e_k = w_k m$ with $m$ fixed, the expenditure elasticity is the same:
	\begin{equation}\label{eq:expelast}
		\frac{\partial \ln e_k}{\partial \ln P_k} \;=\; (1-\sigma)(1-w_k) \;=\; 1-\eta_k ,
	\end{equation}
	where $\eta_k \equiv -\partial\ln Q_k/\partial\ln P_k$ is the uncompensated own-price elasticity of composite demand. Rearranging \eqref{eq:expelast} yields the identity on which the results of this section rest:
	\begin{equation}\label{eq:key}
		\eta_k - 1 \;=\; (\sigma-1)(1-w_k),
		\qquad\text{equivalently}\qquad
		\eta_k \;=\; \sigma(1-w_k)+w_k .
	\end{equation}
	
	Proposition~\ref{prop:characterization} requires $\eta_k\ge1$. Since $w_k\in(0,1)$ implies $1-w_k>0$, \eqref{eq:key} shows that this holds if and only if $\sigma\ge1$. The threshold is therefore independent of the budget share, although the margin by which it is satisfied is not, as discussed below.
	
	\begin{proposition}[Rigidity, not separability, makes the budget markup-proof]\label{prop:rigidity}
		Consider a consumer who budgets in two stages with the CES upper stage \eqref{eq:ces}, and suppose the seller in question is the sole supplier of group $k$, so that $P_k$ is his posted price and $e_k(P_k)$ in \eqref{eq:bk} is his revenue. Then:
		\begin{enumerate}
			\item[(i)] \emph{(Substitutes, $\sigma>1$.)} $\eta_k = 1+(\sigma-1)(1-w_k)>1$, expenditure $e_k(P_k)$ is strictly decreasing, and the consumer's spending rule is markup-proof. The seller's revenue is maximized at the clearing price, and $e_k(P_k)\to0$ as $P_k\to\infty$.
			\item[(ii)] \emph{(Cobb--Douglas, $\sigma=1$.)} $\eta_k=1$ and $e_k = \alpha_k m$ for every $P_k$. The rule is markup-proof, on the boundary of the characterization, and pointwise maximal in the sense of Proposition~\ref{prop:minimal}. The two-stage consumer is observationally indistinguishable from a rigid mental accountant.
			\item[(iii)] \emph{(Complements, $\sigma<1$.)} $\eta_k = 1-(1-\sigma)(1-w_k)<1$, expenditure $e_k(P_k)$ is strictly increasing, and the rule is not markup-proof. The seller's revenue is strictly increasing in price at every price, with supremum equal to the consumer's entire income:
			\[
			\sup_{P_k>0} e_k(P_k) \;=\; \lim_{P_k\to\infty} e_k(P_k) \;=\; m .
			\]
		\end{enumerate}
		A rigid mental account, which holds $e_k$ fixed at $b_k$ by construction, yields $\varepsilon=1$ and is markup-proof for every $\sigma$, and caps the seller's revenue at $b_k$.
	\end{proposition}
	
	Part (iii) makes the magnitude of the exposure explicit. A consumer with complementary spending groups who re-optimizes rationally presents the group-$k$ monopolist with a revenue function that is increasing in price throughout, and the only bound on what he collects is her income. The rigid account, by contrast, bounds his revenue at $b_k$, so the amount at stake in the decision to fix the budget is $m-b_k$, the expenditure the re-optimizing consumer remains willing to redirect toward group $k$. CES preferences admit no choke price, which is why income rather than a reservation value provides the bound; a finite reservation value would truncate the seller's price at that value, leaving the direction of the incentive unchanged. Cases~(ii) and~(iii) also locate the two commitment forms of Section~\ref{sec:background} within a single family: the Cobb--Douglas case reproduces the mental accountant exactly, and the quantity-target buyer arises as a limiting case of complementarity, as the following corollary records.
	
	\begin{corollary}[The Leontief limit is the quantity-target buyer]\label{cor:leontief}
		As $\sigma\to0$ the upper stage becomes Leontief and \eqref{eq:key} gives $\eta_k\to w_k$. For a category with a small budget share, composite demand is close to perfectly inelastic and group expenditure rises almost proportionally with price. The consumer of Section~\ref{sec:mental} who forms quantity targets rather than budgets thus corresponds to the $\sigma\to0$, small-$w_k$ limit of this family.
	\end{corollary}

	Exercise~\ref{ex:ces} in Appendix~\ref{app:ces-example} works through a two-group example with complements, computing the seller's price, the quantity he withholds, and the utility the re-optimizing consumer loses relative to a rigid account.
	
	\subsection{Why substitutes generate markup-proofness and complements do not}
	
	The identity \eqref{eq:key} decomposes the condition into two factors, both of which must be present to generate markup-proofness.
	
	The first factor, $\sigma-1$, measures whether the consumer's substitution response is strong enough to reduce her expenditure on a group whose price has risen. An increase in $P_k$ affects $e_k = P_kQ_k$ through two channels: the price paid per unit of the composite rises, raising expenditure, while the quantity purchased falls, lowering it. The second channel dominates only when $\eta_k>1$, and substitutability across groups is what makes the quantity response large enough for it to do so.
	
	The second factor, $1-w_k$, measures the share of expenditure allocated outside the category, and therefore the scope available for reallocation. As $w_k\to1$ the category absorbs the entire budget and \eqref{eq:key} yields $\eta_k\to1$ for every $\sigma$: absent an alternative destination for the funds, even a consumer with a high elasticity of substitution behaves as though she held a fixed budget. Re-optimization requires both a willingness to reallocate expenditure and somewhere to reallocate it.
	
	Complementarity is the configuration a seller would prefer. When groups are complements, the goods in group $k$ raise the marginal value of expenditure elsewhere, so the consumer values her purchases across groups only in combination. An increase in $P_k$ raises the cost of that combination, and her optimal response is to increase expenditure on group $k$ while reducing it elsewhere. Her expenditure on group $k$ then rises with its price, which is precisely the response from which a monopolist profits. The mental accountant declines to reallocate, so her expenditure remains constant and the seller gains nothing from raising price.
	
	This inverts the standard interpretation of the fungibility violation. It is the re-optimizing consumer's flexibility that the seller exploits, while the mental accountant's refusal to revise her budget amounts to a refusal to respond, and expenditure that does not respond cannot be marked up. The rigidity is therefore not a cost incurred alongside the benefit but the source of the benefit itself.
	
	We close this section pointing out two further and important implications of our analysis. First, when groups are strict substitutes ($\sigma>1$) the rational consumer's expenditure is strictly decreasing in $P_k$, so her rule is markup-proof but not pointwise maximal. By Proposition~\ref{prop:minimal} she consumes strictly less at high prices than a fixed budget would deliver, so demand more elastic than unit-elastic removes the markup at a cost in consumption. The rigid mental account is the boundary case that attains the bound.
	
	Second, the rational and behavioral accounts of budgeting can now be compared on a common footing. Both produce a group-level dollar budget at the point of sale, and what determines the elasticity a within-group seller faces is the price-responsiveness of that budget rather than the cognitive process that generates it.
	
	\section{The endogenous choice of a budget}\label{sec:optimality}
	
	The results so far take the consumer's commitment as given. We now ask when she would choose to make it, which requires weighing the markup she avoids against the flexibility she surrenders.
	
	Consider a consumer with income $m$ facing two categories, $k=1,2$. She is small relative to each market, so she takes prices as given. Let $p$ denote the competitive price she faces in either category when her spending rule is markup-proof.
	
	Her per-unit values are uncertain when she commits. With equal probability the state is $A$, in which her values are $(v_H,v_L)$, or $B$, in which they are $(v_L,v_H)$, with $v_H>v_L>0$. Write $\bar v = \tfrac12(v_H+v_L)$ for her mean value and $\Delta = v_H-\bar v = \tfrac12(v_H-v_L)>0$ for its dispersion.
	
	Under \emph{commitment} she fixes budgets $(b_1,b_2)$ with $b_1+b_2=m$ before the state is realized, and by Proposition~\ref{prop:characterization} faces the clearing price $p$ in each category. Under \emph{no commitment} she keeps the flexibility to direct her whole income to whichever category turns out to be high-value, but she presents a quantity-target rule, which is not markup-proof, and the seller marks price up to some $p^N>p$. She spends her entire income either way, so the two regimes can be compared by the gross value of the goods she obtains.
	
	\begin{proposition}[When budgeting is optimal]\label{prop:optimality}
		Under commitment the consumer's optimal budgets are symmetric, $b_1=b_2=m/2$, and every split yields the same expected payoff $\bar v\,m/p$. She prefers to commit if and only if
		\[
		\frac{p^N}{p} \;>\; 1+\frac{\Delta}{\bar v},
		\]
		that is, if and only if the markup she avoids exceeds one plus the coefficient of dispersion of her values. The value of commitment is increasing in the seller's market power $p^N/p$ and decreasing in the dispersion $\Delta$.
	\end{proposition}
	
	The previous result makes explicit the fact that the trade-off is between two distinct benefits. Rigidity costs the consumer the option to concentrate her spending where it turns out to be most valuable, a loss proportional to $\Delta$. It buys her the difference between the competitive and the marked-up price, a gain proportional to $p^N/p$. She budgets when market power is large relative to the dispersion of her values, and remains flexible otherwise. The structure of this trade-off parallels the commitment-versus-flexibility tension in \citet{galperti2019theory}, with market power in place of the self-control problem as the source of the demand for commitment.

\paragraph{Costly attention and the value of commitment.}
Proposition~\ref{prop:optimality} isolates the market-power motive for committing by charging non-commitment only the cost of forgone flexibility, $\Delta$. \citet{kHoszegi2020choice} identify a second cost of not committing: tracking prices well enough to redirect income optimally is itself costly. The two motives combine additively in the consumer's decision.

Suppose that under no commitment, learning which category has realized the high value $v_H$ requires the consumer to pay an attention cost $\kappa\ge0$, denominated in the same units as her payoff, while committing requires no such cost since her split $b_1=b_2=m/2$ is optimal regardless of the state. An uncommitted consumer who does not pay $\kappa$ can do no better than guess the state, and guessing is dominated by committing whenever the seller has genuine market power, so the only choice worth comparing to commitment is paying $\kappa$ to learn the state and reoptimize.

\begin{proposition}[The value of commitment under costly attention]\label{prop:attention}
Under the stated attention cost, the consumer prefers to commit if and only if
\[
\frac{p^N}{p} \;>\; 1+\frac{\Delta}{\bar v} - \frac{p^N\kappa}{m\bar v}.
\]
The right-hand side is strictly decreasing in $\kappa$: any positive attention cost makes commitment easier to justify than under Proposition~\ref{prop:optimality}, and the two thresholds coincide at $\kappa=0$.
\end{proposition}

Market power and attention cost act as substitutes on the margin of whether to commit. A consumer facing weak market power may still rigidify if tracking prices is sufficiently onerous, and a consumer with negligible attention costs may still rigidify if market power is sufficiently strong; a consumer facing both frictions needs less of either one to justify budgeting than she would if she faced only one. This makes formal the complementarity between the two rationales described in Section~\ref{sec:related}: rather than two competing explanations for the same behavior, they are two additive terms in a single commitment decision.
	
	\paragraph{Comparative statics in market structure.}
	Proposition~\ref{prop:optimality} takes $p^N$ as exogenous. Making it a function of the intensity of competition lets the value of budgeting be evaluated as a function of market structure.
	
	Suppose $n\ge1$ symmetric sellers compete in the category the consumer does not commit to, with $n$ treated as a continuous index of competitive intensity. Take the markup implied by Cournot competition with linear demand, where the Lerner index scales as $1/n$:
	\[
	p^N(n) \;=\; p + \frac{v_H-p}{n}.
	\]
	At $n=1$ this recovers Proposition~\ref{prop:optimality} exactly, with $p^N(1)=v_H$. As $n\to\infty$, $p^N(n)\to p$.
	
	The competitive price $p$ does not require an $n$ subscript: under commitment the buyer is budget-capped, and Proposition~\ref{prop:aggregation} establishes that no seller gains from marking up against a budget-capped buyer, regardless of the number of competitors he faces. The comparative static therefore operates entirely through the no-commitment side of the comparison.
	
	\begin{proposition}[Budgeting and market concentration]\label{prop:concentration}
		Let $p^N(n) = p + (v_H-p)/n$ for $n\ge1$, and suppose $\bar v>p$. Then there is a threshold
		\[
		n^\ast \;=\; \frac{(v_H-p)\,\bar v}{p\,\Delta}
		\]
		such that the consumer prefers to commit if and only if $n<n^\ast$. The value of commitment, $V(n) = \bar v m/p - v_H m/p^N(n)$, is strictly decreasing in $n$, strictly positive at $n=1$, and satisfies $V(n)\to -m\Delta/p<0$ as $n\to\infty$.
	\end{proposition}
	
	The incentive to form a tight category budget is therefore strongest where competition is weakest, and vanishes as the category approaches perfect competition. A standard concentration measure such as the Herfindahl--Hirschman index is proportional to $1/n$ for symmetric firms, so the model delivers a testable prediction: mental-accounting behavior should be more common, and category budgets tighter, in more concentrated categories, and should weaken as local competition intensifies.
	
	A fully microfounded version would replace the reduced-form markup with an explicit oligopoly model, possibly with private information on both sides. The qualitative comparative static does not depend on that choice, since it follows from any markup function decreasing in $n$ that converges to $p$.
	
	\section{Related literature}\label{sec:related}
	
	We now place our models and results in the context of existing work on mental accounting, buyer-optimal demand and buyer commitment, separability, and pricing to budget-constrained buyers.
	
	\subsection{Mental accounting.} \citet{thaler1985mental} and \citet{thaler1999mental} introduce mental accounting and document its departures from fungibility. \citet{shefrin1988behavioral} develop the planner-doer model in which mental accounts serve as a self-control device. \citet{koch2016goals} model narrow psychological accounts as commitments against a consumer's own future temptation. \citet{galperti2019theory} derives good-specific budgeting from the joint demand for commitment (against self-control problems) and flexibility (against uncertainty in intratemporal trade-offs). \citet{prelec1998red} model the hedonic coupling between the pleasure of consumption and the pain of paying. In each of these accounts the friction being solved is internal to the consumer---a conflict between present and future selves, or between certainty and flexibility---and none considers the consumer's problem facing a seller with market power.
	
	\subsection{Budgeting in dollars or in units.} \citet{kHoszegi2020choice} is the closest work to our framework. They show that a consumer who finds it costly to track prices adopts a rigid, category-level dollar budget once the relevant elasticity of substitution crosses a threshold (one-half in their baseline quadratic specification), and that once adopted, her category expenditure is \emph{exactly} constant across realized prices, not merely less responsive (their Proposition 3 and Corollary 1). That object is mechanically identical to ours: exactly constant expenditure is exactly the $\varepsilon=1$ case of Proposition~\ref{prop:characterization}, so a Kőszegi-Matějka hard budget is markup-proof in our sense as soon as it is adopted. The two papers therefore agree entirely on what rigidity delivers, and differ on why a consumer would want it. Their elasticity threshold governs whether a consumer with a fixed cost of attention chooses to become rigid at all, in place of reoptimizing every period at no cost; our threshold in Proposition~\ref{prop:rigidity} instead asks, of a consumer who already reoptimizes for free, whether her ordinary, price-responsive demand happens to be markup-proof anyway. Hence, the market-power motive we identify operates on the same margin as their attention-cost motive---whether to commit at all---rather than on the margin of what a given commitment delivers once made. Our result in Proposition~\ref{prop:attention}  makes the two motives jointly formal: each lowers the bar the other must clear on its own, so a consumer facing weak market power may still rigidify if attention is sufficiently costly, and a consumer with negligible attention costs may still rigidify if market power is sufficiently strong. 
	\smallskip

	A related literature studies how a buyer's limited attention shapes a seller's pricing. \citet{ravid2020ultimatum} shows that a rationally inattentive buyer facing a take-it-or-leave-it offer earns a strictly positive payoff, and \citet{cusumano2024competing} study price competition among sellers facing a rationally inattentive consumer. In these papers the buyer's leverage comes from how she processes information; in ours it comes from the rigidity of her budget, and Proposition~\ref{prop:attention} shows that an attention cost and market power enter the decision to commit additively.
	
	\subsection{Buyer-optimal demand.} The closest antecedent to Section~\ref{sec:background} is \citet{condorelli2020information}, in which a buyer chooses the distribution of her valuation before a monopolist makes a take-it-or-leave-it offer. The buyer-optimal distribution generates unit-elastic demand, the seller is indifferent across an interval of prices, and trade occurs with probability one. \citet{roesler2017buyer} obtain the same unit-elastic structure when the buyer instead chooses what to learn about a fixed valuation, and the equal-revenue distribution plays a related role in the buyer-optimal segmentation of \citet{bergemann2015limits}. Proposition~\ref{prop:characterization}, Corollary~\ref{cor:nodwl}, and Proposition~\ref{prop:minimal} should be read as the budget-commitment counterpart of this result. There the buyer commits to a value distribution and here to a spending rule, and in both cases the seller's incentive to mark up disappears because the buyer's expenditure does not rise with price. The contribution of the present note lies in the object through which the commitment is made---a dollar budget, the object the mental-accounting literature documents---and in embedding it in a multi-category consumer problem, which is what permits the comparison between rigid and re-optimized budgets in Section~\ref{sec:gorman} and the comparative statics of Section~\ref{sec:optimality}.
	
	\subsection{Buyer commitment against sellers.} A broader literature studies how buyers obtain better terms by limiting their own responsiveness. \citet{dana2012buyer} shows that buyer groups, by committing to purchase exclusively from one seller, intensify price competition even when their members have heterogeneous preferences. \citet{burkett2015endogenous,burkett2016optimally} studies principals who impose budget constraints on the bidders acting for them; there the budget resolves an agency problem between principal and bidder rather than a pricing problem between buyer and seller, which is the margin studied here.
	
	\subsection{Separability and two-stage budgeting.} \citet{gorman1959separable} and \citet{deaton1980economics} establish the separability conditions under which two-stage budgeting is valid. A subsequent literature, including \citet{blackorby1978duality} and the revealed-preference characterizations of weak separability in \citet{quah2012revealed} and \citet{cherchye2015semi}, studies when separability can be tested nonparametrically from demand data. That literature treats separability as a property of preferences to be identified or exploited computationally, and does not connect it to the pricing power of sellers within a group. Proposition~\ref{prop:rigidity} indicates why the distinction matters: separability governs whether the group budget is a sufficient statistic, while markup-proofness turns on whether that budget responds to prices. The elasticity on which that proposition rests, $\eta_k=\sigma(1-w_k)+w_k$, has the same form as the share-dependent demand elasticity in nested-CES models of oligopoly with a unit-elastic upper tier \citep{atkeson2008pricing}, and the requirement that this elasticity exceed one for a monopolist's markup to be finite is standard in that literature. Our use of it differs: we read the threshold $\sigma\ge1$ as the condition under which a re-optimizing consumer's group budget is markup-proof, and contrast it with a rigid budget, which is markup-proof for every $\sigma$.
	
	\subsection{Pricing to budget-constrained buyers.} The demand-side mechanism we build on---that budget-constrained rather than quantity-constrained buyers remove a seller's incentive to mark up price---follows from the revenue-optimality result for competitive pricing in budget-constrained Fisher markets established by \citet{finster2025pricing}, and is applied to advertising exchanges (where advertisers declare budgets rather than click quotas) in \citet{melo2026hotw}. To our knowledge no existing paper connects this pricing-side mechanism to mental accounting or to Gorman--Deaton separability; the contribution of this note is that connection.
	
	\section{Discussion}\label{sec:discussion}
	
	This note offers a market-structure rationale for category-level dollar budgeting, a behavior that has been explained almost entirely by appeal to the consumer's internal psychology. The rationale is complementary to the existing ones: a consumer with severe self-control problems and a consumer with none can both benefit from mental accounts, for different reasons, and a consumer facing both frictions has more reason still.
	
	In Appendix~\ref{app:mechanism-design} we derive a mechanism-design implication sharpens the point: a revenue-maximizing seller who could design what buyers commit to would prefer quantity targets, while a buyer who anticipates this prefers a dollar budget. Mental accounting and two-stage budgeting may be read, in this light, as a demand-side response to the seller's side of that problem, adopted because the dollar budget is the buyer-optimal commitment to make.
	
	Proposition~\ref{prop:rigidity}  suggests a sharper reading of the behavioral evidence. What the literature records as a failure of fungibility is what makes the commitment binding: a budget that adjusts to prices, as a fully rational separable consumer's would, is markup-proof only when her spending groups are substitutes, while a budget she refuses to revise is markup-proof in every case and caps the seller's revenue at the budget rather than at her income. The observed rigidity is not incidental to the practice but constitutive of its value.
	
	The argument as developed here is largely static. Proposition~\ref{prop:concentration} endogenizes the markup and shows that the incentive to budget is strongest where competition is weakest and vanishes as the category becomes competitive, yielding the testable prediction that mental-accounting behavior should track local market concentration. That result rests on a reduced-form markup rather than a fully specified oligopoly game with private information on both sides. Extending it to such a setting, and to the many-buyer, many-seller case more generally, is a natural next step.
	
	\bibliographystyle{plainnat}
	\bibliography{mental_accounting_note}
	\appendix
	
	\section{Proofs}\label{app:proofs}
	
	\begin{proof}[Proof of Proposition~\ref{prop:characterization}]
		Fix $s>0$ and let $p^c=p^c(s)$. For $p\le p^c$ we have $q(p)\ge s$, so supply binds and $R(p)=p\,s$, which is strictly increasing; the seller therefore never posts a price below $p^c$. For $p\ge p^c$ we have $q(p)\le s$, so demand binds and $R(p)=p\,q(p)=e(p)$. Hence $p^c$ is revenue-optimal if and only if $e(p)\le e(p^c)$ for all $p\ge p^c$.
		
		Suppose first that $e$ is nonincreasing. Continuity of $q$ gives $q(p^c)=s$, so $R(p^c)=p^c s$, and hence $R(p)=ps\le p^cs=R(p^c)$ for $p\le p^c$, while $R(p)=e(p)\le e(p^c)=R(p^c)$ for $p\ge p^c$. Thus $p^c$ is revenue-optimal, and since $s$ was arbitrary, $q$ is markup-proof.
		
		Conversely, suppose $e$ is not nonincreasing, so there exist $p_1<p_2$ in $[0,v]$ with $e(p_2)>e(p_1)$. Set $s=q(p_1)$ and let $p^c=p^c(s)$. Since $q(p_1)\le s$ we have $p^c\le p_1$, and since $q$ is nonincreasing with $q(p^c)\le s=q(p_1)$ and $p^c\le p_1$, it follows that $q(p^c)=q(p_1)$. Therefore
		\[
		e(p^c)=p^c\,q(p^c)=p^c\,q(p_1)\;\le\;p_1\,q(p_1)=e(p_1)\;<\;e(p_2).
		\]
		Because $p_2>p_1\ge p^c$ and $q(p_2)\le q(p_1)=s$, demand binds at $p_2$, so $R(p_2)=e(p_2)>e(p^c)=R(p^c)$. The seller therefore strictly prefers $p_2$ to the clearing price at this supply level, and $q$ is not markup-proof.
		
		Finally, where $q$ is differentiable, $e'(p) = q(p)\bigl(1-\varepsilon(p)\bigr)$ with $\varepsilon(p) = -q'(p)p/q(p)$, so $e'\le 0$ if and only if $\varepsilon\ge1$.
	\end{proof}
	
	\begin{proof}[Proof of Corollary~\ref{cor:nodwl}]
		By the argument above, under a markup-proof rule the seller's revenue is maximized at $p^c(s)$, at which $q(p^c)=s$: the entire capacity is sold, which is the competitive quantity. Since marginal cost is zero, total surplus is maximized when all units are transferred, so the allocative distortion is zero. Under the budget rule $e\equiv b$ on $[0,v]$, so $R(p)=b$ for every $p\in[p^c,v]$.
	\end{proof}
	
	\begin{proof}[Proof of Proposition~\ref{prop:minimal}]
		Let $q$ be markup-proof and affordable, with expenditure $e(p)=p\,q(p)\le b$ for all $p$, and $e$ nonincreasing by Proposition~\ref{prop:characterization}. Since $q(p)=e(p)/p$, at each fixed $p$ the buyer's consumption is increasing in $e(p)$, so a pointwise-maximal rule is one with pointwise-maximal expenditure. Affordability bounds $e(p)\le b$ for every $p$, and the constant function $e\equiv b$ is nonincreasing, hence markup-proof, and attains this bound at every price. Therefore $q^\ast(p)=b/p$ is markup-proof, affordable, and satisfies $q^\ast(p)\ge q(p)$ for all $p\le v$ and any competing markup-proof affordable $q$. Uniqueness: if $q$ attains the bound at every $p\le v$ then $e(p)=b$ for all such $p$, so $q(p)=b/p$.
	\end{proof}
	
	\begin{proof}[Proof of Proposition~\ref{prop:rigidity}]
		\emph{Step 1: expenditure and its elasticity.} Under two-stage budgeting the consumer's group-$k$ expenditure is $e_k = P_kQ_k$. Taking logarithms and differentiating in $\ln P_k$,
		\begin{equation}\label{eq:pf-generic}
			\frac{\partial \ln e_k}{\partial \ln P_k} \;=\; 1 + \frac{\partial \ln Q_k}{\partial \ln P_k} \;=\; 1-\eta_k ,
		\end{equation}
		where $\eta_k = -\partial\ln Q_k/\partial\ln P_k$. Hence $\partial e_k/\partial P_k\le 0$ if and only if $\eta_k\ge 1$. A within-group seller who raises his price raises $P_k$, so by Proposition~\ref{prop:characterization} the consumer's within-group spending rule is markup-proof exactly when group expenditure is nonincreasing in $P_k$, that is, when $\eta_k\ge1$. This step uses only $e_k=P_kQ_k$ and is therefore independent of the CES specialization.
		
		\emph{Step 2: the CES share elasticity.} With the aggregator \eqref{eq:ces}, cost minimization over the upper stage gives the demands \eqref{eq:cesdemand} and hence the shares \eqref{eq:share}. Write $x_k = \alpha_k^{\sigma}P_k^{1-\sigma}$ and $A_{-k} = \sum_{j\ne k}\alpha_j^{\sigma}P_j^{1-\sigma}$, so that $w_k = x_k/(x_k+A_{-k})$ and $\partial\ln x_k/\partial\ln P_k = 1-\sigma$. Then
		\[
		\frac{\partial \ln w_k}{\partial \ln P_k}
		= \frac{\partial \ln x_k}{\partial \ln P_k} - \frac{x_k}{x_k+A_{-k}}\cdot\frac{\partial \ln x_k}{\partial \ln P_k}
		= (1-\sigma)\bigl(1-w_k\bigr),
		\]
		which is \eqref{eq:shareelast}. Since $e_k = w_k m$ with income fixed, $\partial\ln e_k/\partial\ln P_k = \partial\ln w_k/\partial\ln P_k$, and comparing with \eqref{eq:pf-generic} gives
		\[
		1-\eta_k = (1-\sigma)(1-w_k)
		\qquad\Longleftrightarrow\qquad
		\eta_k = 1+(\sigma-1)(1-w_k) = \sigma(1-w_k)+w_k ,
		\]
		which is \eqref{eq:key}. As a check, the compensated own-price elasticity of CES composite demand is $\sigma(1-w_k)$ and CES preferences are homothetic, so the income elasticity is one and the Slutsky decomposition gives the same $\eta_k = \sigma(1-w_k)+w_k$.
		
		\emph{Step 3: the three cases.} Since $w_k\in(0,1)$, we have $1-w_k>0$, so \eqref{eq:key} gives
		\[
		\eta_k \ge 1 \iff (\sigma-1)(1-w_k)\ge0 \iff \sigma\ge1 ,
		\]
		with strict inequality in each direction. Monotonicity of $e_k$ follows from \eqref{eq:bk}: the map $x\mapsto mx/(x+A_{-k})$ is strictly increasing, and $x_k = \alpha_k^{\sigma}P_k^{1-\sigma}$ is strictly decreasing in $P_k$ when $\sigma>1$, constant when $\sigma=1$, and strictly increasing when $\sigma<1$. Hence $e_k$ is strictly decreasing, constant, or strictly increasing in $P_k$ in the three cases respectively.
		
		For the limits, if $\sigma>1$ then $x_k\to0$ as $P_k\to\infty$, so $e_k\to0$. If $\sigma=1$ then $x_k=\alpha_k$ and $w_k = \alpha_k/\sum_j\alpha_j$, which is independent of all prices; normalizing $\sum_j\alpha_j=1$ gives $e_k=\alpha_k m$. If $\sigma<1$ then $x_k\to\infty$ as $P_k\to\infty$, so $w_k\to1$ and
		\[
		e_k(P_k) = m\,\frac{x_k}{x_k+A_{-k}} \nearrow m .
		\]
		Because $e_k$ is the sole supplier's revenue and is strictly increasing with supremum $m$, the seller has a strict incentive to raise price at every price, and by Proposition~\ref{prop:characterization} the rule is not markup-proof.
		
		\emph{Step 4: the rigid account.} A rigid mental account fixes $e_k=b_k$ independently of $P_k$, so $e(p)\equiv b_k$ is constant, $\varepsilon=1$, and by Proposition~\ref{prop:characterization} the rule is markup-proof for every $\sigma$. The seller's revenue is $b_k$ at every price he might post.
	\end{proof}
	
	\begin{proof}[Proof of Corollary~\ref{cor:leontief}]
		Let $\sigma\to0$ in \eqref{eq:key}: $\eta_k = \sigma(1-w_k)+w_k \to w_k$. Since $w_k<1$, the limit satisfies $\eta_k<1$, so by Proposition~\ref{prop:characterization} the rule is not markup-proof, and by \eqref{eq:expelast} the expenditure elasticity is $1-w_k>0$. For a category with a small budget share the elasticity $\eta_k\approx0$ and the expenditure elasticity $\approx1$, which are the values of the quantity-target rule $q(p)=S\cdot\mathbf 1\{p<c\}$ of Section~\ref{sec:background} on $p<c$.
	\end{proof}
	
	\begin{proof}[Proof of Proposition~\ref{prop:optimality}]
		Under commitment the consumer faces price $p$ in each category and receives $q_k=b_k/p$. In state $A$ her payoff is $(v_Hb_1+v_Lb_2)/p$ and in state $B$ it is $(v_Lb_1+v_Hb_2)/p$. Averaging over the two equally likely states,
		\[
		\mathbb E[\text{payoff}] \;=\; \frac{1}{2p}\bigl[(v_H+v_L)b_1+(v_L+v_H)b_2\bigr] \;=\; \frac{(v_H+v_L)(b_1+b_2)}{2p} \;=\; \frac{\bar v\,m}{p},
		\]
		which is independent of the split; the symmetric split $b_1=b_2=m/2$ is therefore optimal, and every other split is equally good ex ante. Under no commitment she learns the state and directs all income to the high-value category, obtaining $v_H m/p^N$. Comparing,
		\[
		\frac{\bar v\,m}{p} \;>\; \frac{v_H m}{p^N} \iff \frac{p^N}{p} \;>\; \frac{v_H}{\bar v} \;=\; \frac{\bar v+\Delta}{\bar v} \;=\; 1+\frac{\Delta}{\bar v}.
		\]
		The right-hand side is increasing in $\Delta$ and the condition is easier to satisfy as $p^N/p$ rises, giving the stated comparative statics.
	\end{proof}
	
	\begin{proof}[Proof of Proposition~\ref{prop:attention}]
		Under commitment the consumer's payoff is $\bar v m/p$, unaffected by $\kappa$ since no attention is required. Under no commitment, a consumer who guesses the state, choosing a category independently of its realization, obtains $\tfrac12 v_H m/p^N + \tfrac12 v_L m/p^N = \bar v m/p^N$; since $p^N>p$ whenever the seller has genuine market power, this is strictly less than $\bar v m/p$, so guessing is dominated by committing and is never chosen. A consumer who instead pays $\kappa$ to learn the state obtains $v_H m/p^N - \kappa$, as in Proposition~\ref{prop:optimality} net of the attention cost. Comparing the two remaining options,
		\[
		\frac{\bar v m}{p} \;>\; \frac{v_H m}{p^N} - \kappa
		\;\iff\;
		\frac{\bar v}{p} + \frac{\kappa}{m} \;>\; \frac{v_H}{p^N}
		\;\iff\;
		\frac{p^N}{p} \;>\; \frac{v_H}{\bar v} - \frac{p^N\kappa}{m\bar v}
		\;=\; 1+\frac{\Delta}{\bar v} - \frac{p^N\kappa}{m\bar v},
		\]
		using $v_H/\bar v = 1+\Delta/\bar v$ as in the proof of Proposition~\ref{prop:optimality}. The right-hand side is strictly decreasing in $\kappa$ since $p^N,m,\bar v>0$, and setting $\kappa=0$ recovers the threshold of Proposition~\ref{prop:optimality}.
	\end{proof}

	\begin{proof}[Proof of Proposition~\ref{prop:concentration}]
		Substituting $p^N(n)$ into the commitment condition of Proposition~\ref{prop:optimality},
		\[
		\frac{p^N(n)}{p} > 1+\frac{\Delta}{\bar v} \iff \frac{v_H-p}{np} > \frac{\Delta}{\bar v} \iff n < \frac{(v_H-p)\bar v}{p\Delta} = n^\ast.
		\]
		For monotonicity of $V$, $p^N(n)$ is strictly decreasing in $n$ since $v_H>p$, so $v_H m/p^N(n)$ is strictly increasing in $n$, and $V(n) = \bar v m/p - v_H m/p^N(n)$ is strictly decreasing in $n$. At $n=1$, $p^N(1)=v_H$, so $V(1) = \bar v m/p - m = m(\bar v-p)/p>0$ since $\bar v>p$ by assumption. As $n\to\infty$, $p^N(n)\to p$, so $V(n)\to \bar v m/p - v_H m/p = m(\bar v-v_H)/p = -m\Delta/p<0$, using $v_H-\bar v=\Delta$.
	\end{proof}
	
	\section{Supporting derivations}\label{app:derivations}
	
	\subsection{Single buyer}\label{app:single-buyer}
	
	A seller with zero marginal cost and fixed supply $s$ faces one buyer. If the buyer has a fixed dollar budget $b$ and constant per-unit value $v$, she spends her entire budget whenever price $p\le v$: demand is $q(p)=b/p$. Its elasticity is
	\[
	\varepsilon(p) = -\frac{dq}{dp}\cdot\frac{p}{q} = -\left(-\frac{b}{p^2}\right)\cdot\frac{p}{b/p} = 1 \qquad (p\le v),
	\]
	so revenue $p\,q(p) = b$ is constant in $p$ on this range. The seller's revenue, allowing for the supply cap, is
	\[
	R(p) = p\cdot\min\{q(p),\,s\} =
	\begin{cases}
		p\,s & p\le p^c \quad(\text{supply binds}),\\[2pt]
		b & p^c\le p\le v \quad(\text{budget binds}),\\[2pt]
		0 & p>v \quad(\text{buyer exits}),
	\end{cases}
	\]
	where $p^c = b/s$ is the market-clearing price. $R$ is strictly increasing for $p<p^c$ and constant for $p\in[p^c,v]$.
	
	If instead the buyer has a fixed physical quantity target $S\le s$ and reservation value $c$, demand is the step function $q(p) = S\cdot\mathbf 1\{p<c\}$, with $\varepsilon(p)=0$ for $p<c$. Revenue is $R(p) = p\,S$ for $p<c$ and $R(p)=0$ for $p>c$: strictly increasing up to the reservation, so the seller prefers to raise price to $c$.
	
	Figure~\ref{fig:caps-comparison} plots both cases, using $b=6$, $s=2$, $v=4$ (so $p^c=3$) for the budget-capped buyer and $S=2$, $c=4$ for the quantity-capped buyer.
	
	\subsection{Many buyers}\label{app:many-buyers}
	
	\begin{proposition}[Exclusion never pays under budget caps]\label{prop:aggregation}
		Consider finitely many buyers indexed by $i$, and a seller with zero marginal cost. If each buyer is budget-capped with budget $b_i$ and reservation $v_i$, aggregate revenue at any price above the clearing price equals $\sum_{i:v_i\ge p}b_i$, which is constant in $p$ between consecutive reservations and falls by exactly $b_j$ as $p$ crosses $v_j$. Excluding any buyer is therefore a strict revenue loss, and the clearing price is revenue-optimal. If instead each buyer is quantity-capped with target $S_i$ and reservation $c_i$, aggregate revenue above the clearing price is $p\sum_{i:c_i\ge p}S_i$, and excluding a low-reservation buyer is generally profitable.
	\end{proposition}
	
	\begin{proof}
		With budget caps, buyer $i$ contributes $q_i(p)=b_i/p$ while $v_i\ge p$ and nothing otherwise, so her expenditure is $b_i$ whenever she is active. Aggregate expenditure is thus $E(p) = \sum_{i:v_i\ge p} b_i$, constant on any interval containing no reservation value and dropping by $b_j$ at $p=v_j$. Since each remaining buyer's payment is already pinned at her own budget, buyer $j$'s exit generates no offsetting increase, so $E$ is nonincreasing and strictly lower above $v_j$ than below. Applying Proposition~\ref{prop:characterization} to the aggregate rule gives revenue-optimality of the clearing price.
		
		With quantity caps, each active buyer contributes $S_i$ irrespective of $p$, so aggregate expenditure is $p\sum_{i:c_i\ge p}S_i$, which is strictly increasing in $p$ between consecutive reservations. Raising price past $c_j$ forfeits $c_jS_j$ but raises the price paid by every remaining, still price-inelastic buyer; whenever the latter effect dominates, which occurs generally when the excluded buyer's reservation is low relative to those of the others, exclusion raises revenue.
	\end{proof}
	
	\subsection{Mechanism design}\label{app:mechanism-design}
	
	Proposition~\ref{prop:aggregation} has an implication for the commitment a revenue-maximizing seller would prefer buyers to adopt, were he able to design it. A seller who could induce buyers to report and commit to a fixed physical quantity, rather than a fixed dollar budget, would prefer to do so, since quantity commitments constitute the case in which excluding a marginal buyer and raising price on the remainder is profitable. A buyer who anticipates this prefers to commit to a dollar budget, since a genuine budget commitment forecloses that strategy.
	
	Relabeling a quantity target as a dollar amount evaluated at current prices does not remove the asymmetry. What matters is how the commitment responds as price moves away from its current level, and only a genuine budget commitment requires demand, and hence marginal revenue, to contract as price rises.
	
	\begin{remark}[Robustness to nonlinear pricing]\label{rem:nonlinear}
		The results of Section~\ref{sec:background} restrict the seller to a single linear posted price. For a single buyer with a known budget, this restriction costs the seller nothing. \citet[Proposition~1]{chegale2000} show that when the seller can offer an arbitrary nonlinear tariff, menu of lotteries, or financing scheme to a budget-constrained buyer whose budget is known, the revenue-maximizing mechanism is still a single take-it-or-leave-it price, equal to the smaller of the unconstrained monopoly price and the budget. There is no scope for price discrimination to extract more than linear pricing already delivers, which is precisely the conclusion of Corollary~\ref{cor:nodwl} that revenue is pinned at $b$. The single-buyer case therefore survives the move from posted pricing to fully optimal mechanism design.
		
		Proposition~\ref{prop:aggregation} is more fragile, because it is stated for linear pricing only. Once budgets and valuations are private information held by a population of buyers, \citet{chegale2000} and its extension to general type distributions in \citet{devanurweinberg2017} show that a revenue-maximizing seller does generally discriminate, using menus of lotteries and budget-contingent prices, because private budgets create the same screening problem that private valuations do. This does not overturn the aggregation logic, which concerns only how a seller restricted to a single posted price treats budget- versus quantity-capped buyers, but it means the result should be read as a statement about linear pricing, with the fully optimal many-buyer mechanism left as an open extension.
	\end{remark}
	
	\subsection{When the budget fails to be markup-proof}\label{app:limits}
	
	Two qualifications on the main results deserve statement.
	
	\paragraph{The budget must be tight.}
	
	Suppose the consumer's category budget is $X_k$ but she also has an underlying physical need for $S_k$ units. Her demand is
	\[
	q_k(p) = \min\{X_k/p,\, S_k\},
	\]
	with a kink at $p = X_k/S_k$ (Figure~\ref{fig:caps-mixed}). Below the kink the physical need binds and demand is perfectly inelastic. Above it the budget binds and demand is unit-elastic.
	
	\begin{figure}[h]
		\centering
		\includegraphics[width=0.62\linewidth]{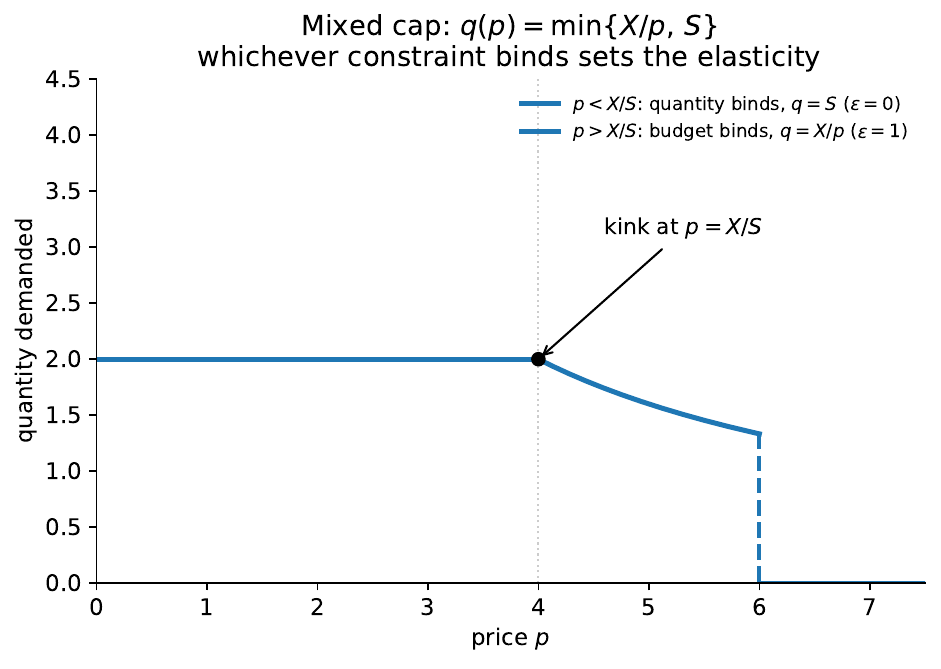}
		\caption{\footnotesize A consumer with both a category budget $X_k$ and a physical need $S_k$: demand $q(p)=\min\{X_k/p,S_k\}$ has a kink at $p=X_k/S_k$. Below the kink the quantity need binds and demand is perfectly inelastic. Above it the budget binds and demand is unit-elastic. See Appendix~\ref{app:mixed}.}
		\label{fig:caps-mixed}
	\end{figure}
	
	\begin{remark}\label{rem:tight}
		Mental accounting makes the consumer's spending rule markup-proof in category $k$ only when the category budget is \emph{tight} relative to her underlying quantity need, that is, only when the budget rather than the physical quantity limits her purchases at the prices she faces. A generous budget, set well above what she would spend on the quantity she wants, provides no markup-proofing.
	\end{remark}
	
	This also indicates where mental accounting should and should not be observed. Discretionary categories with flexible quantities and many potential purchases, such as dining out, entertainment, and clothing, are the ones in which a tight dollar budget can bind before a physical quantity target does. They are also the categories in which mental accounting is most often reported. Categories involving a single large, indivisible commitment, such as a home or an insurance policy, do not offer room for the budget to bind at the margin, and mental accounting in the Thaler sense is less salient there. Two-stage budgeting still describes the allocation across such groups at a coarser level.
	
	\paragraph{The seller's technology.}
	
	The results assume zero marginal cost and capacity-constrained supply, which is the case of a seller who wishes to sell his entire endowment. With strictly positive marginal cost the seller's profit from a budget-capped buyer is $b - c\,q(p)$, which is increasing in $p$, so he raises price toward the buyer's reservation value and sells less. Revenue remains pinned at $b$, so he extracts no more, but the buyer receives fewer units.
	
	Budget commitment therefore bounds what a seller can collect under any cost structure. It delivers the full competitive quantity when marginal cost is negligible relative to price, which is the empirically relevant case for capacity services and digital goods.
	
	\subsection{Mixed case derivation}\label{app:mixed}
	
	A buyer with both a budget $X$ and a physical quantity target $S$ has demand $q(p) = \min\{X/p,\,S\}$. For $p<X/S$ the budget would buy more than $S$ units, so the quantity cap binds: $q(p)=S$ and $\varepsilon(p)=0$. For $p>X/S$ the budget buys fewer than $S$ units, so the budget cap binds: $q(p)=X/p$ and $\varepsilon(p)=1$. The two regimes meet at the kink $p=X/S$ (Figure~\ref{fig:caps-mixed}, drawn with $X=8$ and $S=2$, so the kink lies at $p=4$).
	
	Expenditure is $e(p)=pS$ below the kink and $e(p)=X$ above it, so $e$ is increasing on the first segment and constant on the second. By Proposition~\ref{prop:characterization} the rule is markup-proof only if the relevant price range lies above $X/S$. A budget that does not bind at the prices the buyer actually faces provides no markup-proofing, since the quantity cap binds first and demand is inelastic where it matters.
	
	\subsection{A worked example with complements}\label{app:ces-example}
	
	\begin{exercise}\label{ex:ces}
		Let $K=2$, $\sigma=\tfrac12$, $\alpha_1=\alpha_2=\tfrac12$, income $m=100$, and $P_2=1$. Group~1 is supplied by a single seller with zero marginal cost, capacity $s=50$, and a price cap $\bar P=9$, which plays the role of the finite reservation value discussed after Proposition~\ref{prop:rigidity}. Compare a consumer who re-optimizes her budgets in two stages with a rigid mental accountant who fixes $\bar e_1=\bar b_2=50$.
		\begin{enumerate}
			\item[(a)] Show that the re-optimizing consumer's group-1 expenditure is $e_1(P)=100\sqrt P/(\sqrt P+1)$. Compute $e_1$, $w_1$, and $\eta_1$ at $P\in\{1,4,9\}$, and confirm that $\eta_1<1$ throughout.
			\item[(b)] Find the seller's revenue-maximizing price, revenue, and quantity sold against each consumer.
			\item[(c)] Compare the two consumers' utilities.
			\item[(d)] Repeat (a)--(c) with $\sigma=2$.
		\end{enumerate}
	\end{exercise}
	
	\paragraph{Solution.}
	\emph{(a)} With $\alpha_1=\alpha_2$ the factor $\alpha_k^{\sigma}$ cancels in \eqref{eq:bk}, and $P_k^{1-\sigma}=\sqrt{P_k}$, so $e_1(P)=m\sqrt P/(\sqrt P+1)$ with $m=100$. The budget share is $w_1=\sqrt P/(\sqrt P+1)$ and, by \eqref{eq:key}, $\eta_1=\tfrac12(1-w_1)+w_1=\tfrac12(1+w_1)$:
	\[
	\begin{array}{c|ccc}
		P & 1 & 4 & 9\\\hline
		e_1(P) & 50 & 66.7 & 75\\
		w_1 & 1/2 & 2/3 & 3/4\\
		\eta_1 & 3/4 & 5/6 & 7/8
	\end{array}
	\]
	Since $w_1<1$, $\eta_1<1$ at every price, and expenditure rises with price as in Proposition~\ref{prop:rigidity}(iii).
	
	\emph{(b)} Against the re-optimizing consumer, group-1 demand is $Q_1(P)=e_1(P)/P=100/\bigl(\sqrt P(\sqrt P+1)\bigr)$, which equals $s=50$ at $P=1$, so the clearing price is $p^c=1$. Above it demand binds and revenue $e_1(P)$ is strictly increasing, so the seller posts the cap $\bar P=9$, collects $75$, and sells $Q_1=75/9=25/3$ units, leaving $125/3$ units of capacity unsold. Against the rigid accountant, $Q_1=50/P$, the clearing price is again $p^c=1$, and revenue equals $50$ at every $P\in[1,9]$; by Proposition~\ref{prop:characterization} the clearing price is revenue-maximizing, and at it the seller sells his full capacity of $50$ units.
	
	\emph{(c)} With $\sigma=\tfrac12$ the aggregator \eqref{eq:ces} is $U=(\alpha_1Q_1^{-1}+\alpha_2Q_2^{-1})^{-1}$. The re-optimizing consumer faces $(P_1,P_2)=(9,1)$ and spends $(75,25)$, so $(Q_1,Q_2)=(25/3,25)$ and
	\[
	U^{\text{flex}}=\Bigl(\tfrac12\cdot\tfrac{3}{25}+\tfrac12\cdot\tfrac1{25}\Bigr)^{-1}=12.5 .
	\]
	The rigid accountant faces $(1,1)$ and spends $(50,50)$, so $(Q_1,Q_2)=(50,50)$ and $U^{\text{rigid}}=50$. The rigid budget coincides with what the re-optimizing consumer would choose at competitive prices; her willingness to revise it is what lets the seller raise the price, and she loses three quarters of her utility as a result. The seller gains $25$ in revenue, the difference $e_1(\bar P)-\bar e_1$, and the allocative loss is the unsold capacity.
	
	\emph{(d)} With $\sigma=2$, $\alpha_k^{\sigma}=\tfrac14$ cancels and $P_k^{1-\sigma}=1/P_k$, so $e_1(P)=100/(1+P)$, which is strictly decreasing, and $\eta_1=2(1-w_1)+w_1=1.5$ at $P=1$. The seller's revenue above the clearing price $p^c=1$ is decreasing, so he clears at $P=1$, collects $50$, and sells $50$ units. Both consumers then obtain $(Q_1,Q_2)=(50,50)$ and $U=(\tfrac12\sqrt{50}+\tfrac12\sqrt{50})^2=50$. With substitutes, flexibility costs the consumer nothing, as in Proposition~\ref{prop:rigidity}(i).

\end{document}